\documentclass[runningheads]{llncs}
\usepackage{graphicx}
\usepackage{amssymb,amsmath}
\newcommand{\R}{\mathbb{R}}

\usepackage{algorithm,algorithmic}

\begin{document}
\title{Controlled Recurrence of a Biped with Torso}
\titlerunning{Controlled Recurrence of a Biped with Torso}
%
\author{Adrien Le Co\"ent \inst{1} \and
Laurent Fribourg \inst{2}}
\authorrunning{A. Le Co\"ent, L. Fribourg}
%
\institute{Department of Computer Science, Aalborg University \\
Selma Largerl\o fs Vej 300, 9220 Aalborg, Denmark \\
\email{adrien@cs.aau.dk} \and
LSV - ENS Paris-Saclay $\&$ CNRS $\&$ INRIA, University Paris-Saclay \\
91 Avenue du Pr\'esident Wilson, 94230 Cachan, France \\
\email{fribourg@lsv.ens-paris-saclay.fr}}
\maketitle              
\begin{abstract}
We have recently used a symbolic reachability method for controlling the 
stability of special hybrid systems called ``sampled switched systems''. We show here how the method can be extended in order to control the stability of more general hybrid systems with guard conditions and state resets. We illustrate the method through the example of a biped robot with 6 state variables, using a proportional-derivative (PD) controller. More specifically, we isolate a state region $R$ such that, starting from a state located in~$R$ just after a footstep, the PD-control makes the robot state return to $R$ at the end of the following footstep.

\end{abstract}
\section{Introduction}

The study of bipedal robot control has been pioneered by McGeer \cite{McGeer}.
The original model considered in \cite{McGeer} had 4 state variables.
Today the experimental implementations of bipedal robots may have 12 state
variables ~\cite{AHHFMPASG17}. In order to synthesize controllers which are 
{\em correct-by-construction} for such sophisticated robots, we need to 
obtain {\em reduced-order dynamics}. 
This is done by designing outputs and classical controllers driving these 
outputs to zero. The resulting controlled system evolves on a lower dimensional manifold and is described by the hybrid zero dynamics (HZD) governing the 
remaining degrees of freedom~\cite{AHHFMPASG17}. In a second step, a
symbolic method constructs a {\em finite-state abstraction} (see~\cite{Tabuada:2009}) of the HZD, 
then synthesizes a controller enforcing the desired specifications to be satisfied on the full order model.
The interest in itself for constructing such finite-state abstractions has been 
illustrated in \cite{HAT17} where a control correct-by-construction is synthesized, without need for a preliminary step of order reduction, 
in the case of a bipedal model with 4 state variables.

In this paper, we show that an alternative symbolic method can 
be used in order to synthesize
{\em directly} (i.e., without constructing a finite-state abstraction)
a controller for a bipedal model with 6 state variables.
Our  symbolic method consists in
isolating a zone $R$ of the 
6-dimensional state space, and proving $R$ to be a basin of (recurrent) attraction.

The plan of the paper is as follows: in Section \ref{sec:method}, 
we present our symbolic direct method for proving controlled recurrence; in Section \ref{sec:control}, we apply the method to a bipedal model with 6 state variables; we conclude in Section \ref{sec:conc}.


\section{Controlled Recurrence Method}\label{sec:method}

In the context of the biped robot with torso \cite{feng2013biped}, the state $x$ is 
a vector of dimension~6 of the form
$(\dot{\theta}_1, \dot{\theta}_2, \dot{\theta}_3, \theta_1, \theta_2,\theta_3)$,
where $\theta_1$ (resp. $\theta_2$, $\theta_3$) is the angle between the 
stance leg (resp. swing leg, torso) with the vertical.
See Figure \ref{fig:schema}.
\begin{figure}[h]
\centering
\includegraphics[width=0.4\textwidth]{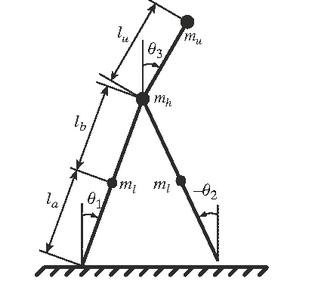}
\caption{Schematic of the robot (taken from \cite{feng2013biped})}
\label{fig:schema}
\end{figure}
The robot is controlled by a proportional-derivative (PD)
controller depending on the angle between the torso and the stance leg (viz., $(\theta_1-\theta_3)$) while the swing leg remains free. We suppose that there is a finite 
set values $U=\{\theta_{SP}^1,\dots,\theta_{SP}^p\}$ for the possible setpoints 
(objective values) assigned to the PD-controller. For each 
value $\theta_{SP}^i$ ($1\leq i\leq p$), the dynamics of the robot is governed by a differential equation of the form $\dot{x}=f_i(x)$. The footstep of the 
robot terminates when there is ``collision'' of the swing leg with the level 
ground, which corresponds to the equation 
\begin{equation}
\theta_1 + \theta_2 = 0.
\label{eq:collision}
\end{equation}
%
At this point, a {\em reset} of the robot
state is performed instantaneously, and a new 
footstep starts. This footstep is governed by an equation of the form 
$\dot{x}=f_j(x)$, where $j$ is the new selected control index; and so on 
iteratively.

The control problem consists in selecting at each collision, an
appropriate control index $i\in\{1,\dots,p\}$ which makes the robot perform
a new footstep (i.e., reach a state  with $\theta_1+\theta_2=0$).
Such a problem can be solved using a {\em controlled recurrence} procedure
(see, e.g., \cite{fs-book13}):

\begin{enumerate}
\item {\em isolation}: isolate a 
rectangular zone $R$ (corresponding to the ``recurrent zone''). 
\item {\em bisection}: divide the zone $R$ into $2^{nD}$ rectangular tiles 
(or ``boxes'') of the same size,\footnote{Actually, the boxes are not {\em all} of the same size, but are generated according to an {\em adaptative} tiling procedure (see Section \ref{ss:proc}).} where~$n$ is the dimension of the state space, and~$D$ the depth of bisection of~$R$. 
\item {\em controlled recurrence}:
for each tile $T$, try to find $i\in\{1,\dots,p\}$ such that, for {\em any} starting point in $T$, the trajectory governed by $\dot{x}=f_i(x)$,  reaches a collision hyperplane (here: $\theta_1+\theta_2=0$), which, after reset, belongs to~$R$.
\end{enumerate}
If, for some tile $T$, the search for an appropriate index $i\in\{1,\dots,p\}$ fails, one can restart the procedure with an incremented value of $D$ or an augmented set of setpoint values $\theta_{SP}^i$.

The controlled recurrence (item 3) is guaranteed using the method of {\em reachability with zonotopes} \cite{alt08,girard05}. 
The initial tile $T$ is seen as a
``zonotope'' \cite{kuhn1998zonotope}. We then compute 
an approximative form of successive discrete-time integrations of~$T$
for $\dot{x}=f_i(x)$, under the form of zonotopes. Let $h$ be the step size of the discrete-time integration sequence.
Let us first suppose  that $f_i$ is {\em linear} of
the form $f_i(x)=Ax+\theta_{SP}^i\ b$. Each $k$-th integration of~$T$ can be computed efficiently in an exact manner, using the structure of {\em zonotopes} (see \cite{girard05,girard2008zonotope}). 
The computation then stops at first step~$k$, say $N^-$, for which the $k$-th image of~$T$ intersects with $\eqref{eq:collision}$ (see Section~\ref{ss:model}). 
As explained in \cite{girard05}, it is easy to compute a lower bound  $\tau^-=N^-h$ of the first time 
(resp. upper bound~$\tau^+=N^+h$ of the last time) for which the intersection of the $k$-th image of~$T$ with hyperplane $\eqref{eq:collision}$ is non-empty. It is also possible to compute an overapproximation of the
{\em continuous} image of~$T$ during time  $t\in[\tau^{-},\tau^{+}]$ intersected with \eqref{eq:collision}
(see, e.g., \cite{girard2008zonotope}); this image 
is denoted
by $Post_{i}^{N}(T)\cap \eqref{eq:collision}$ . The reset mapping 
due to the collision with the floor is then applied, and the resulting set denoted by $Reset(Post_{i}^{N}(T)\cap \eqref{eq:collision})$. The controlled recurrence is guaranteed if, for each tile~$T$ of $R$, one
can find an index $i\in\{1,\dots,p\}$ satisfying: 
\begin{equation}
Reset(Post_{i}^{N}(T)\cap \eqref{eq:collision})\subseteq R.
\label{eq:reset}
\end{equation}

When~$f_i$ is a {\em non-linear} mapping, it is explained, e.g. in
\cite{alt08,girard05,girard2008zonotope}, how to extend the computation
of the image $Post_{i}^{N}(T)$ using zonotopes.
The basic idea is, in our context, to replace the 
nonlinear equation $\dot{x}=f_i(x)$ with an equation of the form 
\begin{equation}
\dot{x}=Ax+\theta_{SP}^i\ b+Hd
\label{eq:robot}
\end{equation}
where $Ax+\theta_{SP}^i\ b$ is the linearized form of $f_i(x)$, $H$ a constant matrix, and $d$ a
``perturbation'' term corresponding to the linearization error. 
It is supposed furthermore that an
upper bound~$\delta$ of~$\|d\|$ is known or can be evaluated 
(see Section \ref{ss:lin} for details).

In Figures \ref{fig:final_sims1}-\ref{fig:final_sims2}, 
the recurrence box $R$ is represented in color cyan, the initial tile $T$
as well as the successive discrete-time integration images
are in blue, and the zone obtained after the
reset operation, is represented in red. One can see that the final red zone is 
located inside $R$ (recurrence property). 
Figures \ref{fig:final_sims3}-\ref{fig:final_sims4} depict an analogous
behavior, but starting from another initial tile $T'$ of $R$.
We explain in further details in Section \ref{sec:control}
how such Figures are generated.

\begin{figure}[h]
\centering
\begin{tabular}{c}
\includegraphics[width=1.0\textwidth,clip,trim = 0cm 6cm 0cm 6cm]{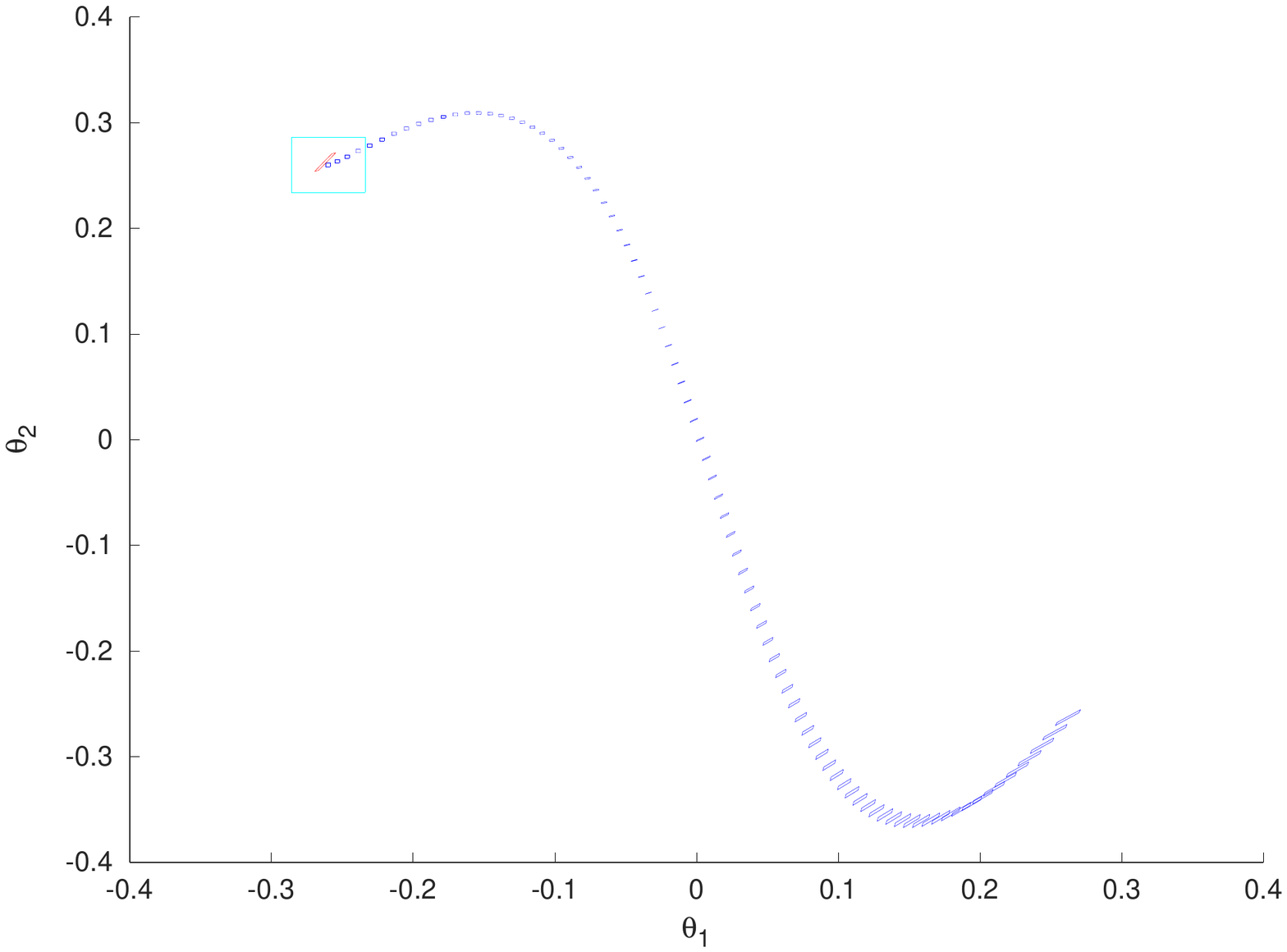} \\
\includegraphics[width=1.0\textwidth,clip,trim = 0cm 6cm 0cm 6cm]{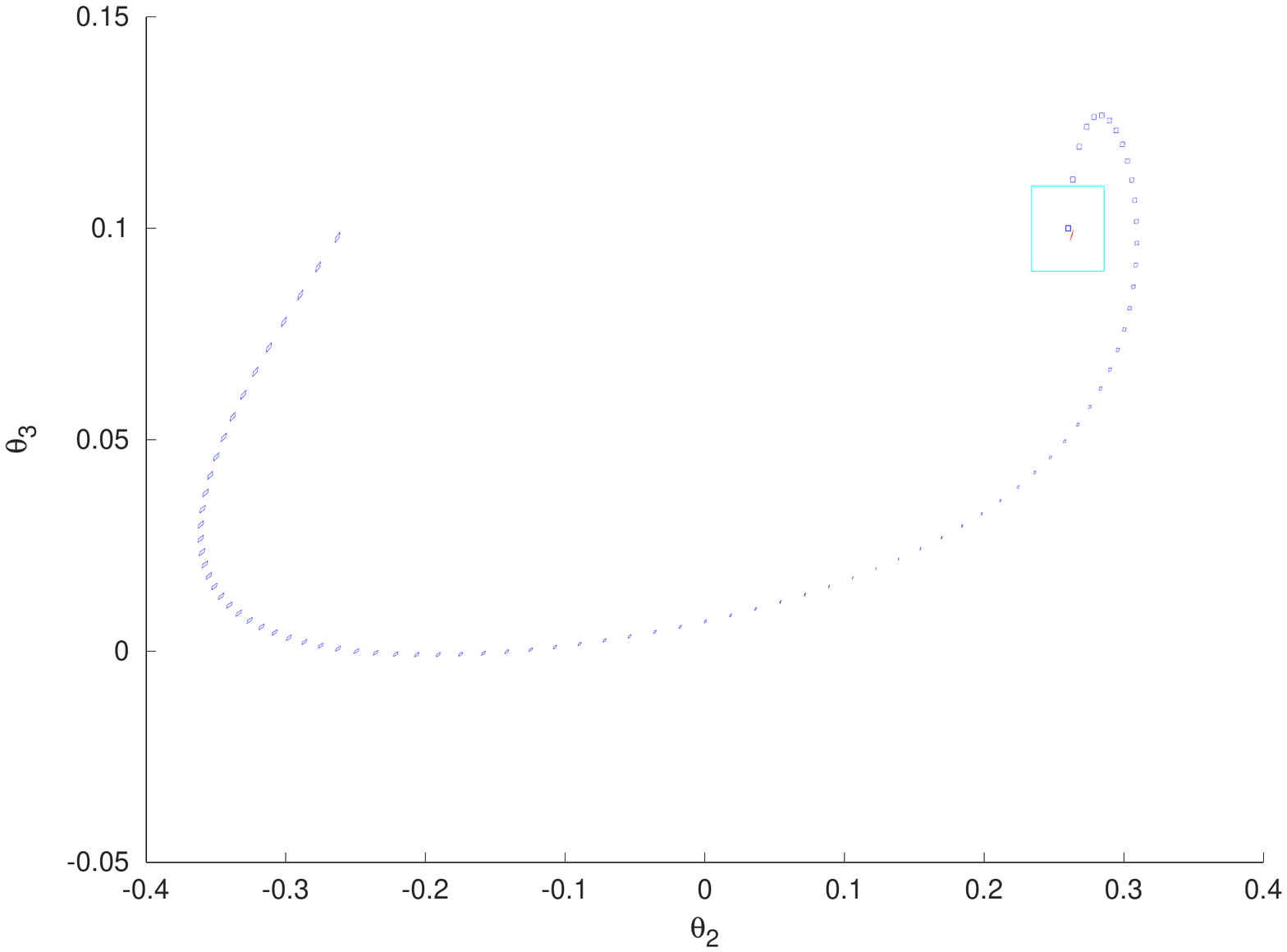} 
\end{tabular}
\caption{$Post^k_i(T)$ in the planes $(\theta_1,\theta_2)$ and $(\theta_2,\theta_3)$ 
for $K_p = 124.675$,
$K_d = 19.25$ and $\theta_{SP}^i = -0.075$.
The cyan boxes correspond to the projections of box $R$. The blue zones are the successive (projections of the)
images $Post^k_{i}(T)$ at discrete times, starting from
$T=[0.58263,0.59737]\times [0.273,0.287]\times [1.36144,1.37856]\times [-0.26162,-0.258375]\times [0.258375,0.26162]\times[0.099375,0.10063]$ located inside $R$. 
The red zones correspond to the final zonotopes, after the reset has been applied. }
\label{fig:final_sims1}
\end{figure}
\begin{figure}[h]
\centering
\begin{tabular}{c}
\includegraphics[width=1.0\textwidth,clip,trim = 0cm 6cm 0cm 6cm]{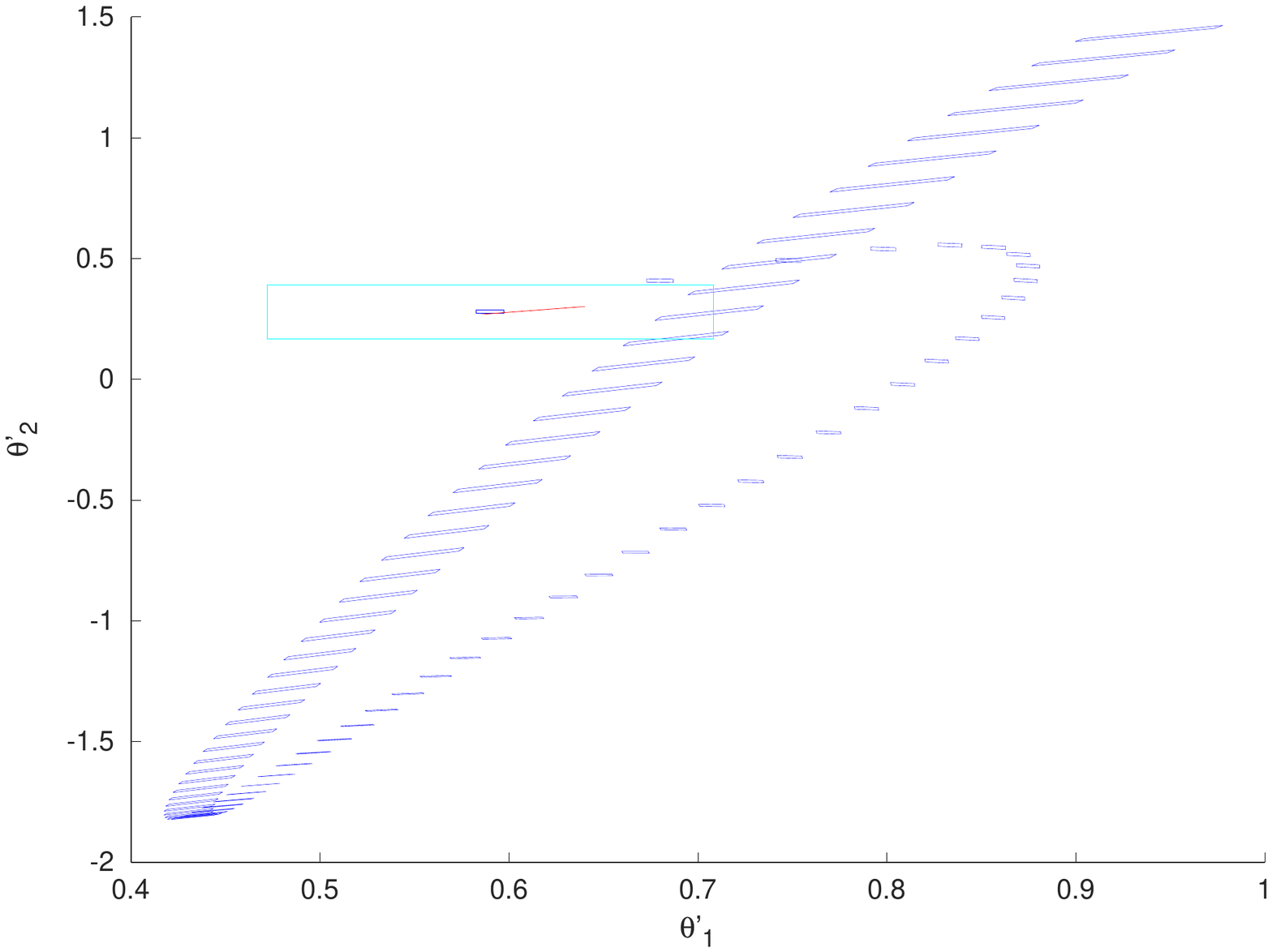}\\
\includegraphics[width=1.0\textwidth,clip,trim = 0cm 6cm 0cm 6cm]{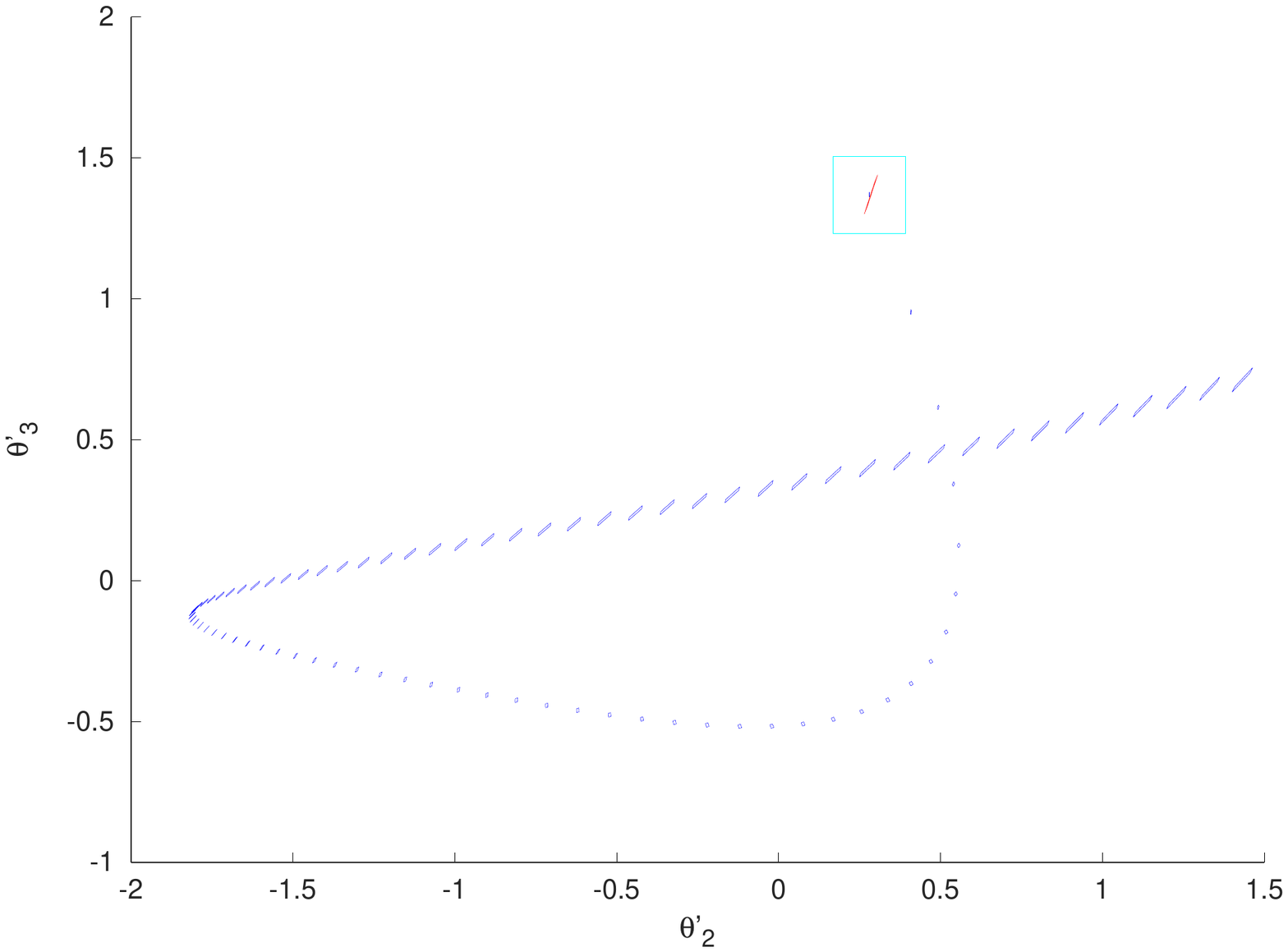}
\end{tabular}
\caption{$Post^k_i(T)$ in the planes $(\dot{\theta_1},\dot{\theta_2})$ and $(\dot{\theta_2},\dot{\theta_3})$ 
for $K_p = 124.675$,
$K_d = 19.25$ and $\theta_{SP}^i = -0.075$.
The cyan boxes correspond to the projections of box $R$. The blue zones are the successive (projections of the)
images $Post^k_{i}(T)$ at discrete times, starting from 
$T=[0.58263,0.59737]\times [0.273,0.287]\times [1.36144,1.37856]\times [-0.26162,-0.258375]\times [0.258375,0.26162]\times[0.099375,0.10063]$ located inside $R$. 
The red zones correspond to the final zonotopes, after the reset has been applied. }
\label{fig:final_sims2}
\end{figure}

\begin{figure}[h]
\centering
\begin{tabular}{c}
\includegraphics[width=1.0\textwidth,clip,trim = 0cm 6cm 0cm 6cm]{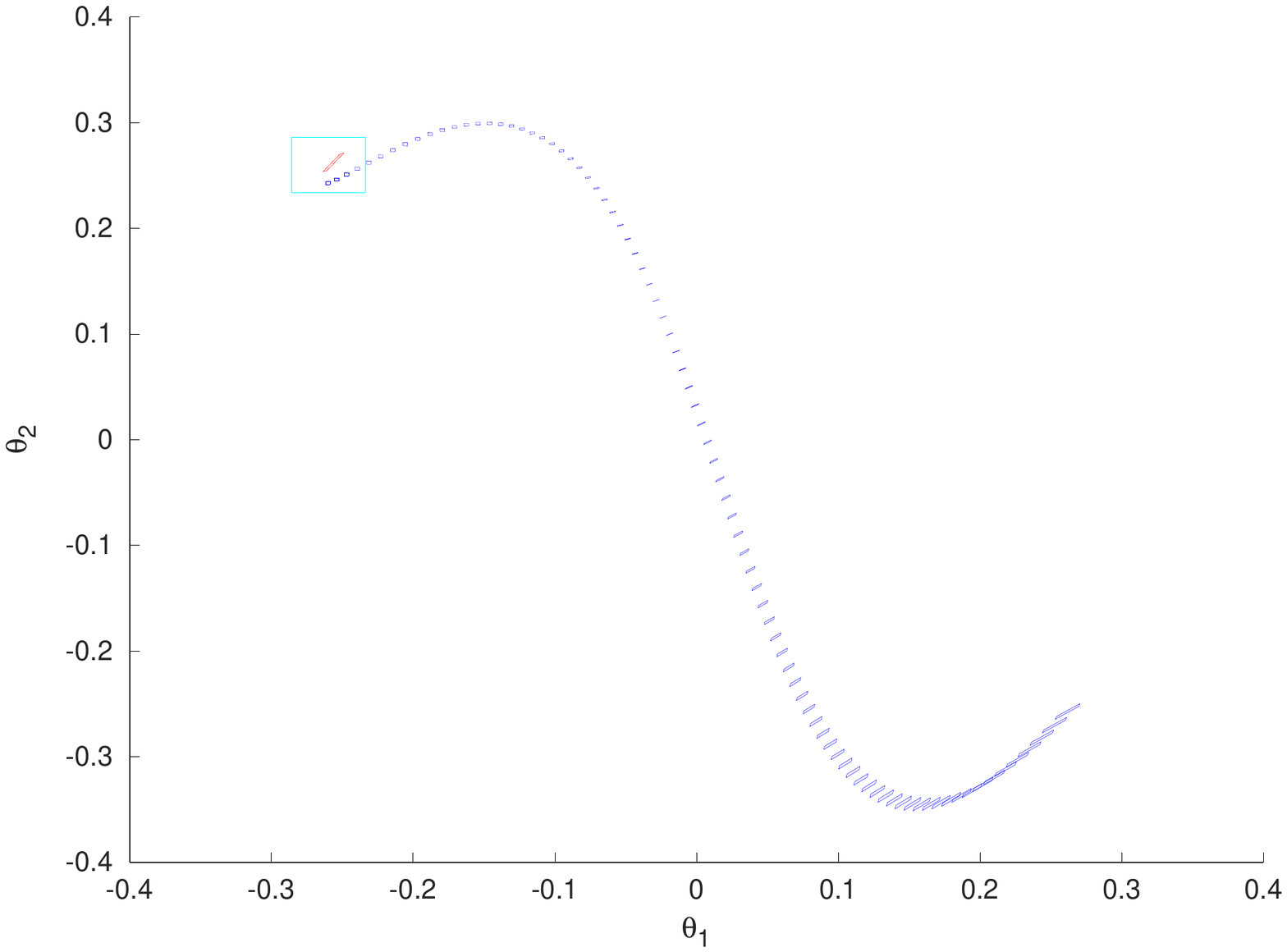} \\
\includegraphics[width=1.0\textwidth,clip,trim = 0cm 6cm 0cm 6cm]{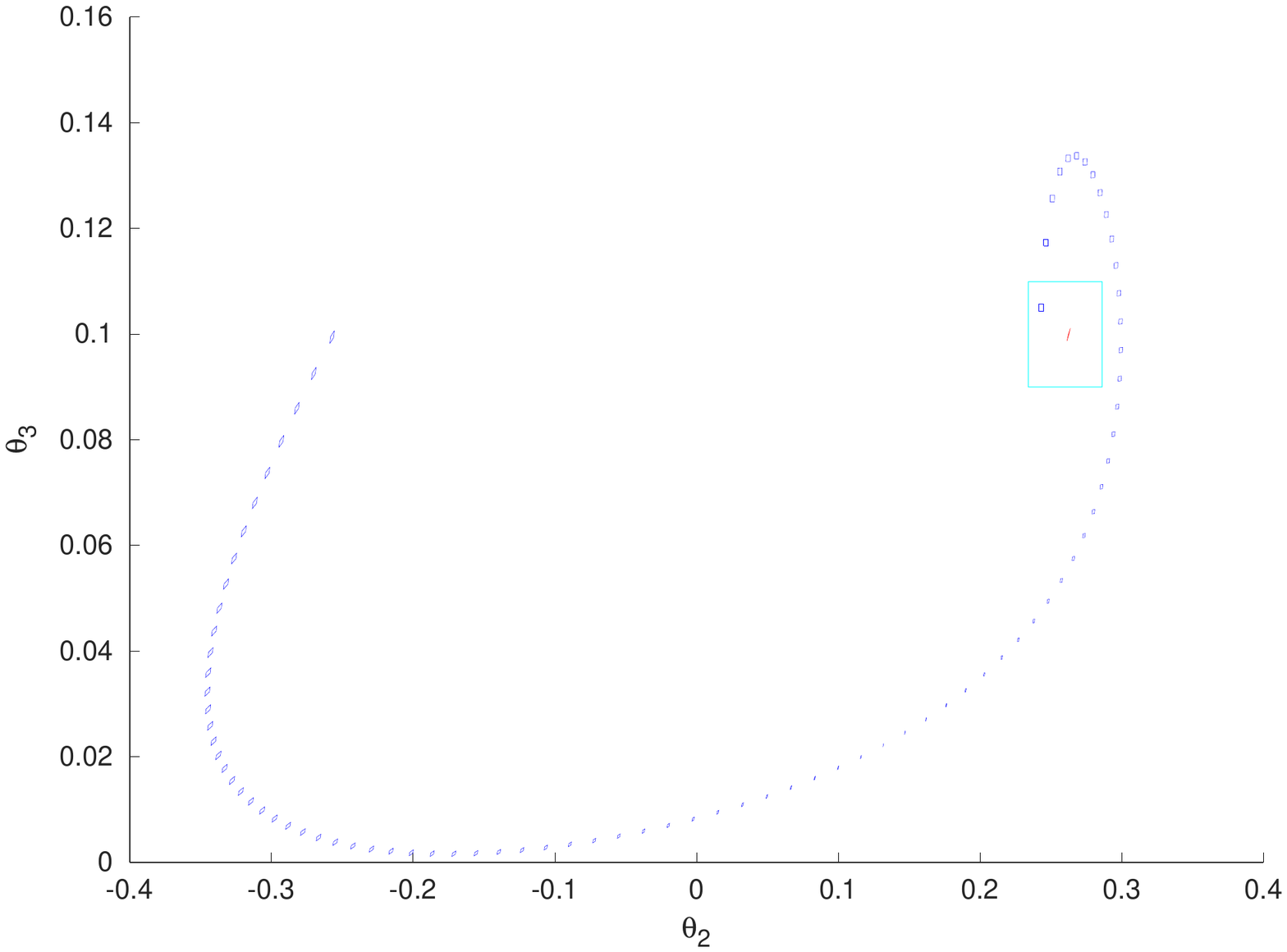} 
\end{tabular}
\caption{$Post^k_i(T')$ in the planes $(\theta_1,\theta_2)$ and $(\theta_2,\theta_3)$ 
for $K_p = 124.675$,
$K_d = 19.25$ and $\theta_{SP}^i = -0.075$.
The cyan boxes correspond to the projections of box $R$. The blue zones are the successive (projections of the)
images $Post^k_{i}(T')$ at discrete times, starting from 
$T'=[0.55300,0.56700]\times [0.2535,0.26650]\times [1.45087,1.46913]\times [-0.24452,-0.24148]\times [0.24218,0.24452]\times [0.10434,0.10566]$ located inside $R$. 
The red zones correspond to the final zonotopes, after the reset has been applied. }
\label{fig:final_sims3}
\end{figure}
\begin{figure}[h]
\centering
\begin{tabular}{c}
\includegraphics[width=1.0\textwidth,clip,trim = 0cm 6cm 0cm 6cm]{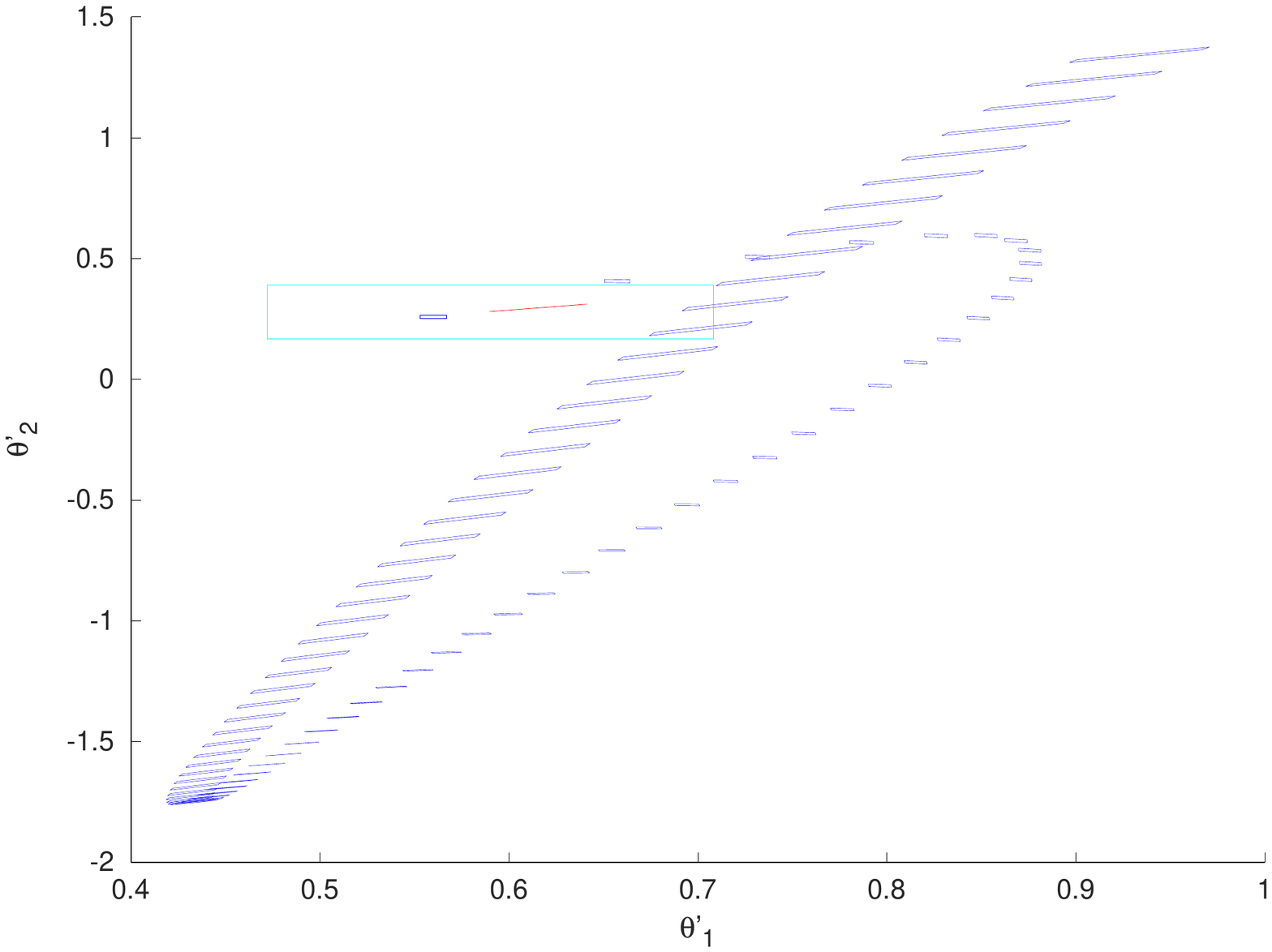}\\
\includegraphics[width=1.0\textwidth,clip,trim = 0cm 6cm 0cm 6cm]{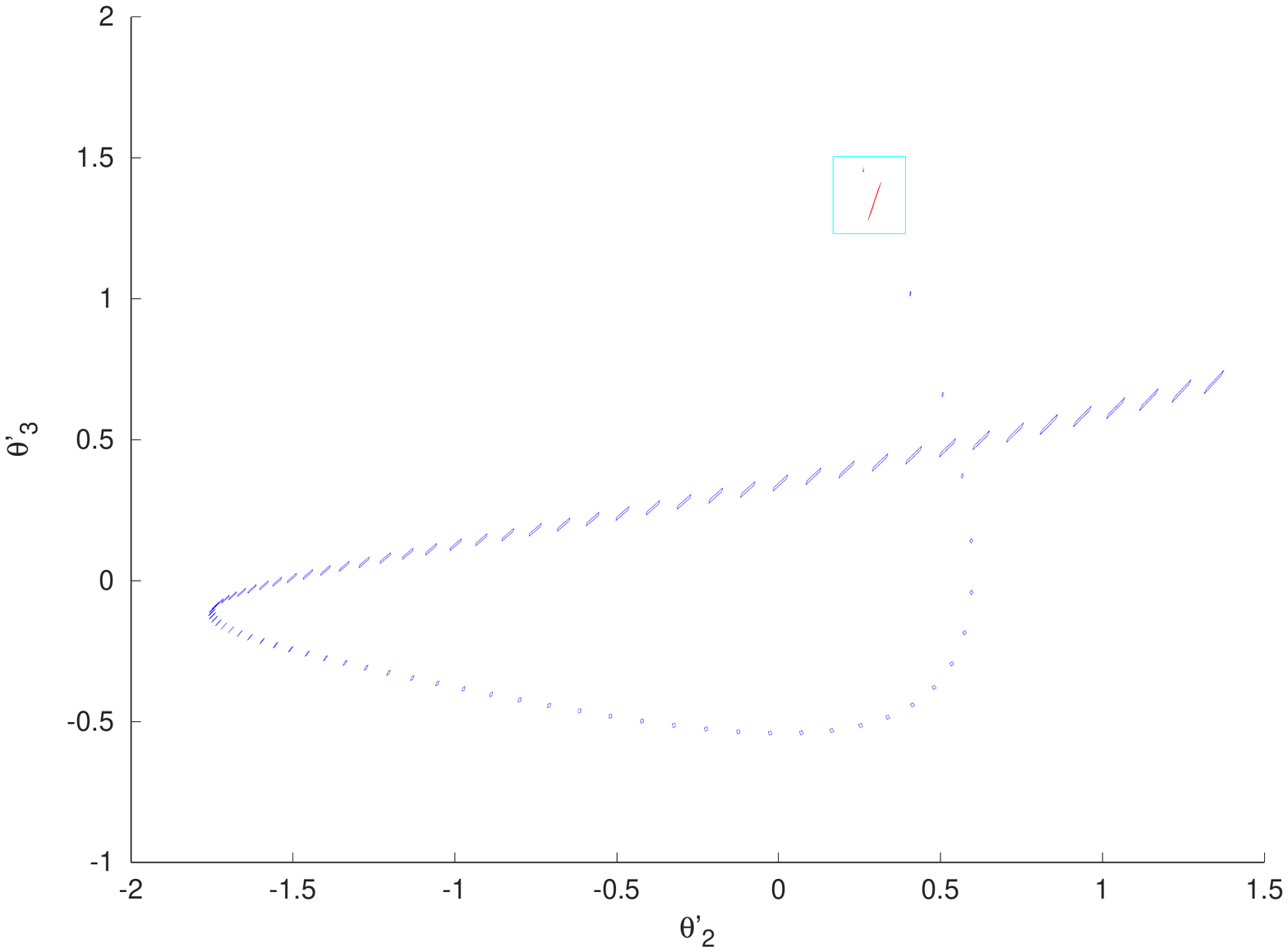}
\end{tabular}
\caption{$Post^k_i(T')$ in the planes $(\dot{\theta_1},\dot{\theta_2})$ and $(\dot{\theta_2},\dot{\theta_3})$ 
for $K_p = 124.675$,
$K_d = 19.25$ and $\theta_{SP}^i = -0.075$.
The cyan boxes correspond to the projections of box $R$. The blue zones are the successive (projections of the)
images $Post^k_{i}(T')$ at discrete times, starting from 
$T'=[0.55300,0.56700]\times [0.2535,0.26650]\times [1.45087,1.46913]\times [-0.24452,-0.24148]\times [0.24218,0.24452]\times [0.10434,0.10566]$ located inside $R$. 
The red zones correspond to the final zonotopes, after the reset has been applied. }
\label{fig:final_sims4}
\end{figure}
%
%
%
 \section{Application to the Biped with Torso}\label{sec:control}
 \subsection{Model}\label{ss:model}
The model is taken from \cite{feng2013biped}, to which the following text is 
mainly borrowed.
The dynamics of the robot consists of a swing phase starting with both feet touching the ground. 
A torque is applied between the stance leg and the torso, so the swing leg moves forward.
It is followed by a collision when both legs touch the ground again, meaning the end of a step.
Due to inherent symmetries in the robot, one can consider that once a step is finished, the previous stance leg becomes the new swing 
leg. This requires the application of a reset in the model (``collision'').
A collision happens when both feet touch the ground. The condition to be met 
for 
the collision is \eqref{eq:collision}: $\theta_1 + \theta_2 = 0$
\footnote{Condition \eqref{eq:collision} is true a first time when the legs are parallel, but we ignore such a ``scuffing'' and assume the swing leg to continue without collision.}.
%
%

Once the collision happens, conservation of the momentum and considering of symmetries in the system 
leads to a {\em reset} to apply, leading to a new set of initial conditions. The equations of the reset are the following:
\[
\begin{pmatrix}
I & 0 \\ 0 & L^n(\theta^+)
\end{pmatrix}
\begin{pmatrix}
\theta^+ \\ \dot \theta^+
\end{pmatrix}
= 
\begin{pmatrix}
Q & 0 \\ 0 & L^o(\theta^-)
\end{pmatrix}
\begin{pmatrix}
\theta^- \\ \dot \theta^-
\end{pmatrix}
\label{eq:collision2}
\]
where $I$ is the identity matrix, and
$$ Q = \begin{pmatrix}
0 & 1 & 0 \\
1 & 0 & 0 \\
0 & 0 & 1
 \end{pmatrix} $$
$$ L^n(\theta) = \begin{pmatrix}
L^n_{11}  L^n_{12} L^n_{13} \\
L^n_{21}  L^n_{22} L^n_{23} \\
L^n_{31}  L^n_{32} L^n_{33} 
 \end{pmatrix} $$

$$ L^o(\theta) = \begin{pmatrix}
L^o_{11}  L^o_{12} L^o_{13} \\
L^o_{21}  L^o_{22} L^o_{23} \\
L^o_{31}  L^o_{32} L^o_{33} 
 \end{pmatrix} $$
 with

 $L^n_{11} = -(m_h+m_l+m_u)(l_a+l_b)^2-m_l l_a^2 -+m_l(l_a+l_b)l_bcos(\theta_1^+-\theta_2^+)-m_u(l_a+l_b)l_ucos(\theta_1^+-\theta_3^+)$,
 $L^n_{12} = m_l l_b (l_a cos(\theta_1^+ - \theta_2^+) - l_b  + l_b  cos(\theta_1^+ - \theta_2^+))$,
 $L^n_{13} =  - m_l l_u (l_u +l_a cos(\theta_1^+ - \theta_3^+)+l_b cos(\theta_1^+ - \theta_3^+))$,
 $L^n_{21} = -m_u(l_a+l_b)l_ucos(\theta_1^+-\theta_3^+)$, 
 $L^n_{22} = 0$,
 $L^n_{23} = -m_u l_u^2$,
 $L^n_{31} = m_l (l_a + l_b)l_b cos(\theta_1^+-\theta_2^+)$,
 $L^n_{32} = - m_l l_b^2$,
 $L^n_{33} = 0$,
 
  $L^o_{11} = m_l l_a l_b -(m_h+m_u)(l_a+l_b)^2    cos(\theta_1^--\theta_2^-) - 2m_l (l_a+l_b)l_bcos(\theta_1^--\theta_2^-) -m_u(l_a+l_b)l_ucos(\theta_1^--\theta_3^-)$,
 $L^n_{12} = m_l l_a l_b$,
 $L^n_{13} =  - m_u l_u (l_u +l_a cos(\theta_2^+ - \theta_3^+)+l_b cos(\theta_2^+ - \theta_3^+))$,
 $L^n_{21} = -m_u(l_a+l_b)l_ucos(\theta_1^+-\theta_3^+)$, 
 $L^n_{22} = 0$,
 $L^n_{23} = -m_u l_u^2$,
 $L^n_{31} = m_l l_a l_b$,
 $L^n_{32} = 0$,
 $L^n_{33} = 0$.\\


A PD-controller controls the torque during the swing phase.
During a step, the dynamics of the robot is given by the nonlinear equation
\begin{equation}
M(\theta)\ddot \theta + N(\theta,\dot \theta) + G(\theta) = c\ u
\label{eq:robot2}
\end{equation}
with $\theta = (\theta_1,\theta_2,\theta_3)^\top$. The vector $u$ is the control input corresponding to a PD-controller defined by:
 \begin{equation}
 u = K_p (\theta_{SP} - (\theta_3 - \theta_1)) - K_d (\dot \theta_3 - \dot \theta_1)
\label{eq:robot2bis}
 \end{equation}
 where $\theta_{SP}$ is the ``setpoint''. In our context, $\theta_{SP}$ belongs to a finite set $U$ of values. The value of $\theta_{SP}$ in $U$ is chosen after each collision for the whole duration of the forthcoming footstep (i.e., till the next collision) in order to make the robot return in the recurrence zone $R$
(see Section \ref{sec:method}). The matrices $M$, $N$, $G$ and $b$ are of the form
$$M(\theta) = \begin{pmatrix}
M_{11}  M_{12} M_{13} \\
M_{21}  M_{22} M_{23} \\
M_{31}  M_{32} M_{33} 
 \end{pmatrix} $$
$$N(\theta, \dot \theta) = \begin{pmatrix}
N_{1} \\
N_{2} \\
N_{3} 
 \end{pmatrix} $$
$$G(\theta) = \begin{pmatrix}
G_{1} \\
G_{2} \\
G_{3} 
 \end{pmatrix} $$
 $$c = \begin{pmatrix}
-1 \\
0 \\
1
 \end{pmatrix} $$
with:

$M_{11} = (m_u +m_h +m_l)(l_a +l_b)^2 +m_ll_a^2$, 

$M_{{12}} = M_{21} = M^*_{12} cos(\theta_1 - \theta_2)$, 
with $M^*_{12}=-m_l(l_a +l_b)l_b$,

$M_{13} = M_{31} = M_{13}^* cos(\theta_1 -\theta_3)$
with $M_{13}^*=m_u(l_a +l_b)l_u$, 

$M_{22} = m_ll_b^2$, 
$M_{23} = M_{32} = 0$, 
$M_{33} = m_ul_u^2$ ; 

$N_1 = N_{12}^*sin(\theta_1 -\theta_2)\dot \theta_2^2 + N^*_{13} sin(\theta_1 - \theta_3)\dot \theta_3^2$, 

$N_{12}^*= -m_l(l_a +l_b)l_b$ and 
$N_{13}^*=m_u(l_a + l_b)l_u$.

$N_2 = N_2^*sin(\theta_1 - \theta_2)\dot \theta_{1}^2$
with $N_2^*=m_l(l_a + l_b)l_b $, 

$N_3 = N_3^*sin(\theta_1 - \theta_3)\dot \theta_{1}^2$
with $N_3^*=-m_u(l_a + l_b)l_u$; 

$G_1 = G_1^*sin(\theta_1)$
with $G_1^*=-((m_h + m_l + m_u)(l_a + l_b) + m_ll_a)g$, 

$G_2 = G_2^*sin(\theta_2)$
with $G_2^*=m_ll_bg$, 

$G_3 = G_3^*sin(\theta_3)$
with $G_3^*=-m_u l_ug$.\\

The linear form of $M(\theta)$ is $M^*\theta$ with:
$$M^* = \begin{pmatrix}
M_{11}  M^*_{12} M^*_{13} \\
M_{12}^*  M_{22} M_{23} \\
M_{13}^*  M_{32} M_{33} 
 \end{pmatrix} $$
Likewise, the linear form of $G(\theta)$ is $G^*\theta$ with
$G^*=(G_1^*,G_2^*,G_3^*)^{\top}$, and the linear form of $N$ is null.
%

A simulation of the PD-controller over two steps with 
$K_p = 124.675$ and
$K_d = 19.25$,
$\theta_{SP} = -0.075$ (for both steps) is given in Figure~\ref{fig:simu}.
\begin{figure}[h]
\centering
\includegraphics[width=0.95\textwidth,clip,trim = 0cm 6cm 0cm 6cm]{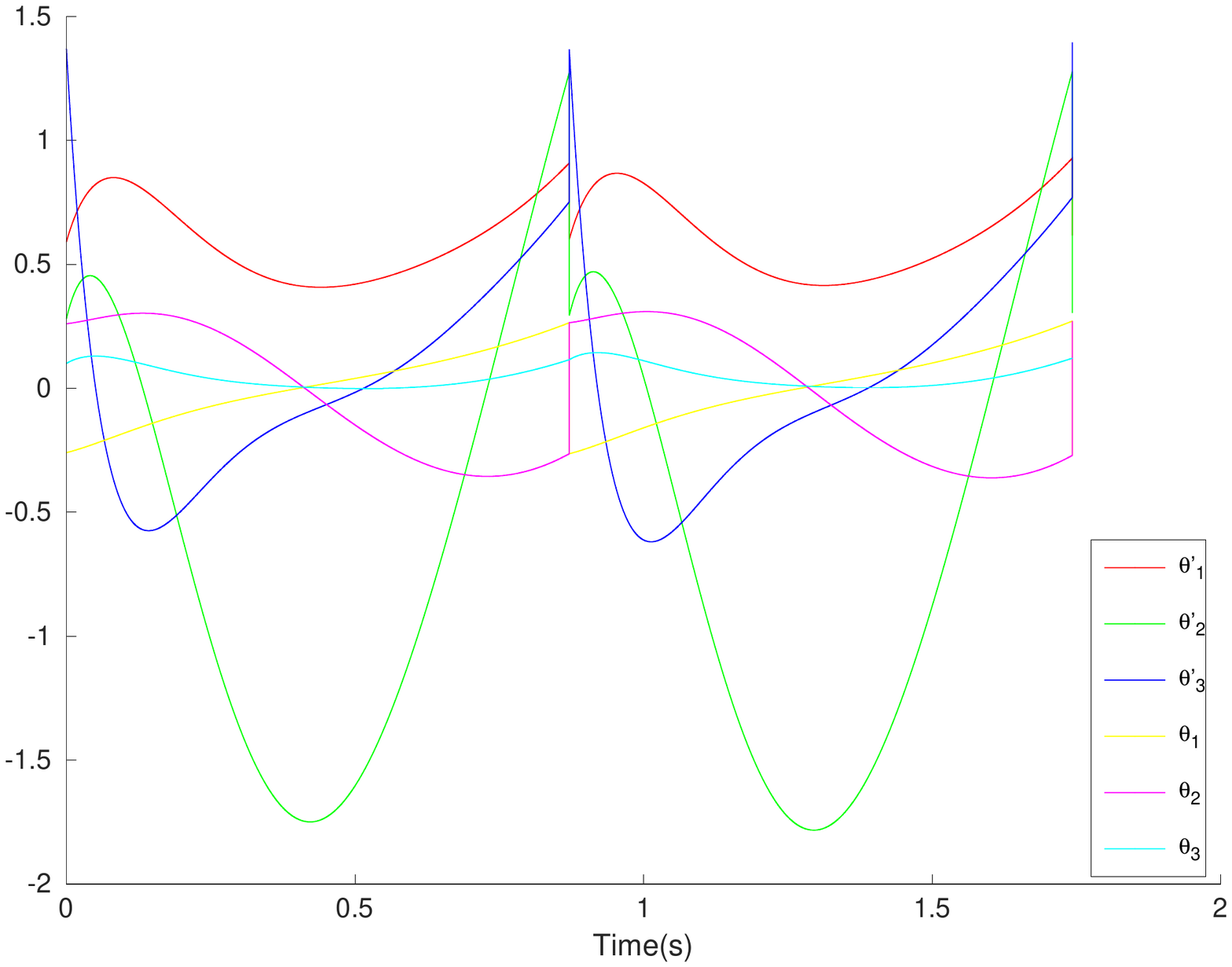}
\caption{Simulation of  two robot footsteps
for $K_p = 124.675$,
$K_d = 19.25$ and
$\theta_{SP} = -0.075$.}
\label{fig:simu}
\end{figure}
\subsection{Linearization with perturbation}\label{ss:lin}
From a general point of view,
we are interested in the control synthesis problem 
for a continuous-time dynamical system subject
to disturbances, described by the set of nonlinear 
ordinary differential equation:
\begin{equation}
  \dot x = f(x,d),
  \label{eq:switched_system0}
\end{equation}
 where $x \in \R^n$ is the state
of the system, and $d \in \R^m$ is a bounded
perturbation. 
The functions $f:\R^n \times \R^m \rightarrow \R^n$, is the vector field
describing the dynamics of the system. There exist today several efficient symbolic tools which perform reachability analysis of nonlinear systems,
and control them in a provably safe manner: eg.,
SpaceEx \cite{SpaceEX},
Flow*
\cite{Flow} or DynIBEX \cite{DynIBEX}.
%
Rather than using such tools, we propose here for the biped case study, to follow a {\em specific linearization} approach 
in order to take advantage of our tool MINIMATOR \cite{fribourg2014finite,fs-book13,minimator}
(cf. Section \ref{ss:proc}).
We first explain how to reformulate 
system~(\ref{eq:robot2}-\ref{eq:robot2bis}) under the linearized form~(\ref{eq:robot}). As mentioned in Section \ref{sec:method}, when the system is under form~\eqref{eq:robot}, one can easily construct (overapproximations of) reachable sets $Post_i^k(T)$
using
zonotopes 
(see, e.g., \cite{girard05}).

\begin{proposition}\label{prop:lin}
System (\ref{eq:robot2}-\ref{eq:robot2bis}) can be written under
the linearized form
with bounded perturbation $d=(d_1,d_2,d_3)^{\top}$:
 \[
 \dot x = A x + \theta_{SP} \ b  + H d
 \label{eq:statespace_error}
 \]
where 
\[
Hd = \begin{pmatrix}
(M^*)^{-1} \begin{pmatrix}
d_1 \\ d_2 \\ d_3
\end{pmatrix}
\\
0 \\ 0 \\ 0
\end{pmatrix}
\]
with, for $i=1,2,3$, $|d_i|\leq \delta_i$ for some constant $\delta_i>0$.
\end{proposition}
\begin{proof}
Given an expression of the form $e(t)$, let us denote by $e_{\max}$ the maximum of $e(t)$ over the $k$-th integration time step (of length $h$):\footnote{The expression $e_{\max}$ differs for each $k$, and the notation should be $e_{\max}^k$, but the upper index $k$ is dropped for the sake of simplicity.}
$$e_{\max} := \max_{t\in[kh,(k+1)h]}e(t).$$
In the context of the biped model, we have:

$G_i=G^*_i\sin(\theta_i)$ satisfies
$G_i=G^*_i\theta_i +d_i^G$ with, for $i=1,2,3$:
$|d_i^G|\leq \delta_i^G$ with 
\[
\delta_i^G := \frac{|G^*_i|}{6} |\theta_i|_{\max}^3.
\]
$N_1=d_1^N=N^*_{12}\sin(\theta_1-\theta_2)\dot{\theta_2}^2
+ N^*_{13}\sin(\theta_1-\theta_3)\dot{\theta_3}^2$ satisfies 
$|d_1^N|\leq \delta_1^N$ with 
\[
\delta_1^N := |N^*_{12}| |\dot{\theta_2}|_{\max}^2|\theta_1-\theta_2|_{\max}
+ |N^*_{13}| |\dot{\theta_3}|_{\max}^2|\theta_1-\theta_3|_{\max}.
\]
$N_2=d_2^N=N^*_2\sin(\theta_1-\theta_2)\dot{\theta_1}^2$ satisfies
$|d_2^N|\leq \delta_2^N$ with 
\[
 \delta_2^N := |N^*_2| |\dot{\theta_1}|_{\max}^2|\theta_1-\theta_2|_{\max}.
\]
$N_3=d_3^N=N^*_3\sin(\theta_1-\theta_3)\dot{\theta_1}^2$ satisfies 
$|d_3^N|\leq \delta_3^N$ with 
\[
 \delta_3^N := |N^*_3| |\dot{\theta_1}|_{\max}^2|\theta_1-\theta_3|_{\max}.
\]
$ M_{12}= M_{21}=M_{12}^*\cos(\theta_1-\theta_2)= M_{12}^* +d_{12}^M$ with:
$|d_{12}^M|\leq \delta_{12}^M$ with 
\[
\delta_{12}^M := \frac{1}{2}|M_{12}^*| (\theta_1-\theta_2)_{\max}^2.
\]
$ M_{13}=M_{31}=M_{13}^*\cos(\theta_1-\theta_3)= M_{13}^* +d_{13}^M$ with:
$|d_{13}^N|\leq \delta_{13}^M$ with 
\[
\delta_{13}^M := \frac{1}{2}|M_{13}^*| (\theta_1-\theta_3)_{\max}^2.
\label{eq:error3}
\]

Let us write 
\[
B^p =
\begin{pmatrix}
-K_p & 0 & K_p \\
0 & 0 & 0 \\
-K_p & 0 & K_p
\end{pmatrix}
\]
and 
\[
B^d =
\begin{pmatrix}
-K_d & 0 & K_d \\
0 & 0 & 0 \\
-K_d & 0 & K_d
\end{pmatrix}
\]

Let us write
\[
A = \begin{pmatrix}
\begin{pmatrix} (M^*)^{-1} B^d \end{pmatrix} &  \begin{pmatrix}(M^*)^{-1}(- G^* + B^p) \end{pmatrix} \\
 \begin{pmatrix}
1 & 0 & 0 \\ 0 & 1 & 0 \\ 0 & 0 & 1 
\end{pmatrix}
 & 
  \begin{pmatrix}
0 & 0 & 0 \\ 0 & 0 & 0 \\ 0 & 0 & 0 
\end{pmatrix} 
\end{pmatrix}
\]
and 
\[
\theta_{SP}\ b = \begin{pmatrix}
(M^*)^{-1} \begin{pmatrix}
- K_p \theta_{SP} \\ 0 \\ K_p \theta_{SP}
\end{pmatrix}
\\
0 \\ 0 \\ 0
\end{pmatrix}
\]

System (\ref{eq:robot2}-\ref{eq:robot2bis}) can thus be reformulated as 
the following linearized system 
with bounded perturbation $d=(d_1,d_2,d_3)^{\top}$ as
 follows:
 \[
 \dot x = A x + \theta_{SP} \ b  + H d
 \label{eq:statespace_error2}
 \]
where 
\[
Hd = \begin{pmatrix}
(M^*)^{-1} \begin{pmatrix}
d_1 \\ d_2 \\ d_3
\end{pmatrix}
\\
0 \\ 0 \\ 0
\end{pmatrix}
\]
with, for $i=1,2,3$, $|d_i|\leq \delta_i$ with:
\[
\delta_1 := \delta_1^G + \delta_1^N + \delta_{12}^M |\ddot{\theta_2}|_{\max}+\delta_{13}^M|\ddot{\theta_3}|_{\max}
\]
%
\[
\delta_2 :=  \delta_2^G + \delta_2^N + \delta_{12}^M |\ddot{\theta_1}|_{\max}
\]
\[
\delta_3 := \delta_3^G + \delta_3^N + \delta_{13}^M |\ddot{\theta_1}|_{\max}
\]
We have thus obtained a system of the form \eqref{eq:robot}
with 
$d=(d_1,d_2,d_3,0,0,0)^T$. Furthermore, we have: $|d_i|\le \delta_i$ for $i=1,2,3$.
\hspace*{\fill} $\Box$
\end{proof}


When we perform discrete-time integration, we will now check  that, at each time
step (of length~$h$), the norm of the perturbation $(M^*)^{-1}d$ is always less than or equal to 
$\frac{1}{10}\| \theta_{SP}^i b\|$. More precisely, 
at each $k$-th step,
we check that the upper bound of the linearization error $\delta :=\max\{\delta_1,\delta_2,\delta_3\}$ satisfies:
\begin{equation}
\delta\leq \frac{K_p}{10} |\theta^i_{SP}|.
\label{eq:error}
\end{equation}
This guarantees that the linearization error (seen as a perturbation) is always ``small''
with respect to the constant term $\theta_{SP}^i b$ of 
\eqref{eq:robot}. Given an initial tile~$T$, we then construct an overapproximation of
$Post_i^k(T)$ 
for $1\leq k\leq N^+$, by 
\begin{enumerate}
\item  computing the images of $T$ through successive
discrete-time {\em linear} integrations,  and
\item extending these images to account for error $\delta$.
\end{enumerate}
Both operations are efficiently performed
using zonotopes.
A similar approach (linearization with addition of a disturbation term) can be done for the collision phase. The sets $Post_i^N(T)$ 
(corresponding to the {\em continuous-time} reachable set $Post_i^t(T)$ for $t\in[N^-h,N^+h]$)
and $Reset(Post_i^N(T))$ (due to the collision phase) are then computed along the lines of the method 
sketched out in Section \ref{sec:method}.

\subsection{MINIMATOR procedure}\label{ss:proc}

 In order to prove the recurrence property
\eqref{eq:reset}, we adapt the MINIMATOR algorithm
defined in~\cite{fribourg2014finite}:
given the input box $R$ and
a positive integer $D$, the algorithm provides, when
it succeeds, a decomposition (by bisection) $\Delta$ of $R$ of the form $\{ (T_i,
\theta_{SP}^i) \}_{i \in I}$, with the properties:
\begin{itemize}
\item $\bigcup_{i \in I} T_i = R$,
\item $\forall i \in I, \ Reset(Post^{N}_{i}(T_i)) \subseteq R$
and, for all $k=1,\dots,N^+$,\  $\delta\leq \frac{K_p}{10}|\theta_{SP}^i|$.
\end{itemize}

Here $R$ is  a ``box'' (i.e., a cartesian product of
real intervals), which is seen here as a special form of zonotope.
The tiles $\{ T_i \}_{i \in I}$ are sub-boxes obtained by repeated
bisection.  At first, function $Decomposition$ calls sub-function
$Find\_Control$ which looks for a value $\theta_{SP}^i \in U$ 
such that $Reset(Post^{N}_{i}(R)) \subseteq R$.  If such a value $\theta_{SP}^i$ is
found, then a uniform control over $R$ is found. 
Otherwise, $R$ is divided into two
sub-boxes $T_1$, $T_{2}$, by bisecting $R$ w.r.t. its longest
dimension. Values of $\theta_{SP}$ are then searched to control these sub-boxes. 
If for each $T_i$, function
$Find\_Control$ manages to get a value for $\theta_{SP}^i$ 
verifying $Reset(Post^{N}_{i}(T_i)) \subseteq R$, then it is
a success and algorithm stops.  
If, for some $T_j$, no such mode is
found, the procedure is recursively applied to $T_j$.  It ends with
success when every sub-box $T_i$ of $R$ has a value of $\theta_{SP}^i$ verifying the latter
conditions, or fails when the maximal degree of decomposition $D$ is
reached.  The algorithmic form of functions $Decomposition$ and
$Find\_Control$, adapted from \cite{fs-book13}, are given in Algorithm~\ref{algo:decomposition} and
Algorithm~\ref{algo:findpattern} respectively. The procedure is called initially
with $Decomposition(R,R,D)$, i.e. $T := R$.

 \begin{algorithm}[ht]
  \centering
  \begin{algorithmic}
    \STATE{\textbf{Function:} $Decomposition(T,R,D)$}
    \STATE{\begin{center}\line(1,0){150}\end{center}}
    \STATE{\quad \textbf{Input:} A box $T$, a box $R$, a degree $D$ of bisection}
    \STATE{\quad \textbf{Output:}$\langle\{(T_i,\theta_{SP}^i)\}_{i},True\rangle$ or $\langle\_ ,False\rangle$}
    \STATE{\begin{center}\line(1,0){150}\end{center}}
    \STATE{ $(\theta_{SP},bool) := Find\_Control(T,R)$}
    \IF{$bool=True$}{
      \STATE{\textbf{return} $\langle\{(T,\theta_{SP})\},True\rangle$}
    }
    \ELSE
    \IF{$D = 0$} \RETURN{$\langle \_,False\rangle$} \ELSE
    \STATE{Divide equally $T$ into $(T_1, T_{2})$ \FOR{$i=1,2$}\STATE{\small{$(\Delta_i,bool_i)$ := $Decomposition(T_i,R,D - 1)$}}\ENDFOR
      \RETURN $(\bigcup_{i=1,2} \Delta_i,\bigwedge_{i=1,2} bool_i)$ } \ENDIF
    \ENDIF
  \end{algorithmic}
  \caption{Algorithmic form of Function $Decomposition$.}
  \label{algo:decomposition}
\end{algorithm}

 \begin{algorithm}[ht]
  \centering
  \begin{algorithmic}
    \STATE{\textbf{Function:} $Find\_Control(T,R)$}
    \STATE{\begin{center}\line(1,0){150}\end{center}}
    \STATE{\quad \textbf{Input:}A box $T$, a box $R$}
    \STATE{\quad \textbf{Output:}$\langle \theta_{SP},True\rangle$ or $\langle\_, False\rangle$}
    \STATE{\begin{center}\line(1,0){150}\end{center}}
    \STATE{$U :=$ finite set of values of $\theta_{SP}^i$}
    \WHILE{$U$ is non empty} \STATE{Select $\theta_{SP}^i$ in $U$}
    \STATE{$U:= U\setminus  \{\theta_{SP}^i\}$}
    \IF{$Reset(Post^{N}_{i}(T)\cap \eqref{eq:collision}) \subseteq R$
with, for all $1\leq k\leq N^+$, $\delta\leq \frac{K_p}{10}|\theta_{SP}^i|$}{\RETURN{$\langle \theta_{SP}^i,True\rangle$}} 
\ENDIF
    \ENDWHILE
    \RETURN{$\langle \_,False \rangle$}
  \end{algorithmic}
  \caption{Algorithmic form of Function $Find\_Control$.}
  \label{algo:findpattern}
\end{algorithm}


\subsection{Experimentation}\label{ss:exp}

The verification procedure is implemented in an adaptation of the tool MINIMATOR \cite{minimator},
using the interpreted language Octave, and the experiments are performed on a 2.80 GHz Intel Core i7-4810MQ CPU with 8 GB of
memory. 

The physical constants of the biped robot
are the same as those given in \cite{feng2013biped}:
$m_u=m_h=10$, $m_l=5$, $l_u=l_b=l_a=0.5$, $g=9.81$.
The  PD-constants used here are: $K_p = 124.675$, $K_d = 19.25$.
The set of possible values for the setpoint $\theta_{SP}$ is
$U=\{-0.07 + 0.001\ i\}_{i=0,\pm 1,\pm 2,\dots, \pm 9}$.

The verification procedure requires 3h hours of computation time.
We use an integration
time step of length $h=10^{-2}s$.
We manage to control entirely the zone 
$R=[0.48,0.72]\times [0.18,0.42]\times [1.26,1.54]\times [-0.286,-0.234]\times [0.234,0.286]\times [0.09,0.11]$
with a bisection depth $D = 4$: for each tile $T$ of $R$ there is some $\theta_{SP}\in U$ which makes $T$ return to $R$.
For example, as illustrated in 
Figures~\ref{fig:final_sims1}-\ref{fig:final_sims2},
the value $\theta_{SP}=-0.075$
makes the initial tile 
$T=[0.58263,0.59737]\times [0.273,0.287]\times [1.36144,1.37856]\times [-0.26162,-0.258375]\times [0.258375,0.26162]\times[0.099375,0.10063]$
return to $R$.
Figures~\ref{fig:final_sims3}-\ref{fig:final_sims4} show that
the same value $\theta_{SP}=-0.075$ achieves also the recurrence property for
tile 
$T'=[0.55300,0.56700]\times [0.2535,0.26650]\times [1.45087,1.46913]\times [-0.24452,-0.24148]\times [0.24218,0.24452]\times [0.10434,0.10566]$\footnote{In Figures \ref{fig:final_sims1}-\ref{fig:final_sims4}, we did not plot {\em all} the images $Post^k(T)$,
and $Post^k(T')$, for $1\leq k\leq N^+$, but only some of them for the sake of readability of the pictures. }.

\section{Final remarks}\label{sec:conc}

We have shown how a direct symbolic method for proving controlled recurrence 
successfully applies to the robot model of \cite{feng2013biped}.
Up to our knowledge, this is the first time that a 
symbolic method has synthesized a control correct-by-construction 
for a bipedal robot of dimension 6.

It would be interesting to try this method 
to  higher dimensional robots or 
other hybrid systems with impact, possibly using
a preliminary step of order reduction along the lines of
\cite{AHHFMPASG17}.
This would probably require to use, not a specific linearization technique as here, but a general procedure designed for 
nonlinear reachability analysis, as in SpaceEx \cite{SpaceEX},
Flow*
\cite{Flow} or DynIBEX \cite{DynIBEX}.

\bibliographystyle{plain}
\bibliography{biblio}
\end{document}